\documentclass[conference]{IEEEtran}

\IEEEoverridecommandlockouts

\usepackage{times}
\usepackage{setspace}
\usepackage{graphicx}
\usepackage{subfig}
\usepackage{amsmath}
\usepackage{amsthm}
\usepackage{comment}
\usepackage{wrapfig}
\usepackage{caption}
\usepackage{float}
\floatstyle{plaintop}
\restylefloat{table}

\newtheorem{theorem}{Theorem}[section]
\usepackage{algorithm}
\usepackage{algorithmic}
\usepackage{amsmath}
\usepackage{mathbbol}
\usepackage{bbm}
\usepackage{graphicx}
\usepackage{comment}
\usepackage{etoolbox}

\newcommand{\ind}[1]{\textbf{I}(#1)}

\newcounter{subeqn}

\newcommand{\argmax}{\arg\!\max}

\newlength{\smallsimfigwidth}
\newlength{\simfigwidth}
\title{MobiCacher: Mobility-Aware Content Caching in Small-Cell
Networks}

\author{
\IEEEauthorblockN{
Yang Guan$^\dagger$ 
~~ Yao Xiao$^\ddagger$ 
~~ Hao Feng$^\ddagger$ 
~~ Chien-Chung Shen$^\dagger$
~~ Leonard J.  Cimini, Jr.$^\ddagger$ 
}
\IEEEauthorblockA{
$\dagger$~Department of Computer and Information Sciences, \\
$\ddagger$~Department of Electrical and Computer Engineering,\\
University of Delaware, USA \\
Email: \{yangg, yxiao, haofeng, cshen, cimini\}@udel.edu}
\thanks{This material is based on research sponsored by the Air Force
Research Laboratory (FA9550-12-1-0086) and NSF (CNS-1016841). The U.S.
Government is authorized to reproduce and distribute reprints for
Governmental purposes notwithstanding any copyright notation
thereon.}}

\begin{document}
\maketitle

\begin{abstract}
  Small-cell networks have been proposed to meet the demand of ever growing mobile data traffic. One of the prominent challenges faced by small-cell networks is the lack of sufficient backhaul capacity to connect small-cell base stations (small-BSs) to the core network. We exploit the effective application layer semantics of both spatial and temporal locality to reduce the backhaul traffic. Specifically, small-BSs are equipped with storage facility to cache contents requested by users. As the {\em cache hit ratio} increases, most of the users' requests can be satisfied locally without incurring traffic over the backhaul. To make informed caching decisions, the mobility patterns of users must be carefully considered as users might frequently migrate from one small cell to another. We study the issue of mobility-aware content caching, which is formulated into an optimization problem with the objective to maximize the caching utility. As the problem is NP-complete, we develop a polynomial-time heuristic solution termed {\em MobiCacher} with bounded approximation ratio. We also conduct trace-based simulations to evaluate the performance of {\em MobiCacher}, which show that {\em MobiCacher} yields better caching utility than existing solutions.

\end{abstract}

\section{Introduction\label{sec:intro}}

We are in the era of a mobile data explosion. Cisco's most recent
Visual Networking Index (VNI) \cite{index2013global} reports that
global mobile data traffic reached 1.5 exabytes per month at the end
of 2013.  In addition, mobile data traffic is forecasted to grow at a
compound annual growth rate of more than 60\% from 2013 to 2018. The
ever growing demand from users of mobile wireless networks, in terms
of both capacity and coverage, is driving wireless technologies ({\em
e.g.}, LTE, LTE-A, IEEE 802.11ac/ad/af, WiMAX, {\em etc}.) alike
towards their limits.  To confront the challenges, one effective
approach to boosting both the spatial reuse and the area spectral
efficiency is to bring transmitters closer to receivers ({\it i.e.},
decrease the distance between transmitters and receivers).  Motivated
by this idea, small-cell networks (SCNs) \cite{hoydis2011green} have
been proposed, where small-cell base stations (small-BSs), including
pico, femto and relay base stations and WiFi access points, are
deployed together with macrocell base stations (Fig.~\ref{fig:scn}).
When combined with widened spectrum and more efficient links, SCNs
have the potential to provide the order-of-magnitude increase in
capacity required \cite{hoydis2011green}.

\begin{figure}[tb]
  \centering
  \includegraphics[width=.33\textwidth]{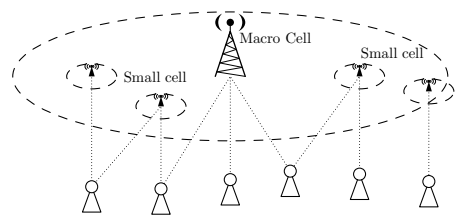}
  \caption{A typical small-cell network\label{fig:scn}}
\end{figure}

The emergence of SCNs, however, poses new challenges for the
operations of cellular networks. First, unlike macrocell base stations
that are connected to the core network via optical fibers, the
backhaul bandwidth of small-BSs is usually constrained.  For instance,
the backhaul of a femtocell base station is typically facilitated by
the customers' home networks such as DSL or cable modem
\cite{damnj2011a}. In addition, due to the reduced cell size and the
large number of small cells, mobile users (MUs) may frequently move
from one cell to another \cite{lopez2011enhanced}, resulting in rapid
fluctuations in the traffic load. 

The limitations on spectrum and per-link spectral efficiency and the
challenges of maximizing spatial reuse in SCNs prompt us to investigate
solutions beyond the physical layer and network topology. In this
paper, we propose to exploit the application layer semantics of both
spatial and temporal locality. Specifically, small-BSs are equipped
with storage facility for caching {\it popular} content in order to
tradeoff against constrained backhaul capacity. 

The benefits of caching popular contents at small-BSs are two-fold.
First, given enough content reuse, most of the requests from MUs can
be fulfilled locally, so that much less overhead will be incurred by
retrieving the requested contents from remote servers over the
backhaul. Second, storing contents at small-BSs reduces the content
retrieval delay. Experiments conducted in \cite{dandapat2013sprinkler}
have shown that downloading contents from a local cache can be up to 8
times faster than that from a remote server. Such reduction in the
retrieval time greatly improves the Quality of Experience (QoE)
perceived by MUs, particularly when MUs are highly mobile and have
short contact time with small-BSs.

However, given that the storage capacity on a small-BS is finite, only
limited contents can be cached. Furthermore, because the ``local''
preference for particular contents of the MUs currently located within
a small cell may be quite different from the ``global'' popularity of
contents, caching decisions made solely based on the global popularity
will be sub-optimal. In addition, as small cells are densely deployed,
a single MU may be served by multiple small-BSs at the same time so
that caching duplicate copies of the same contents at multiple nearby
small-BSs not only wastes the cache capacity but also fails to
diversify the cached contents among small-BSs. To address these
issues, this paper proposes a mobility-aware content caching problem
for SCNs. In particular, we argue that, as to be shown in
Sec.~\ref{sec:motivation}, because MUs may frequently migrate from one
small cell to another, small-BSs shall carefully choose which contents
to cache by taking the mobility patterns of MUs into consideration, in
order to best serve the MUs currently located within the cell of a
small-BS as well as its future MUs.

The paper proceeds in Sec.~\ref{sec:related}
to review related work.  Sec.~\ref{sec:motivation} motivates the
benefits of considering the mobility patterns of MUs when making
caching decisions. Detailed network architecture and problem
formulation are described in Sec.~\ref{sec:framework}. As the problem
is NP-complete, a heuristic solution termed {\em MobiCacher} with
bounded approximation ratio is presented in Sec.~\ref{sec:heuristic}.
We evaluate the performance of {\em MobiCacher} in
Sec.~\ref{sec:results}, and conclude the paper in
Sec.~\ref{sec:conclusion}.

\section{Related Work\label{sec:related}}

To better serve MUs subject to low-capacity backhaul, caching popular
contents at small-BSs was first envisioned in
\cite{femtocaching2012golrezaei}. The work considered a group of
stationary MUs and defined a distributed caching problem, {\em
i.e.}, which content should be cached by which small-BS to
minimize the average network delay. However, as demonstrated in
Sec.~\ref{sec:motivation}, given that MUs might frequently migrate
from one small cell to another so that the traffic load of each small cell
may fluctuate drastically, a caching solution that ignores MUs'
mobility will be sub-optimal.

In \cite{dandapat2013sprinkler}, a distributed video storage system
termed Sprinkler was proposed. Motivated by the observation that the
contact time between an MU and a small-BS is typically short so that
the MU may not be able to download one complete video file from a
single small-BS, Sprinkler breaks a video file into chunks and
distributes the chunks to the series of small-BSs along the MU's
trajectory. The main objective of Sprinkler is to properly cache video
chunks so that a particular video chunk can be fetched by the MU
before the chunk's playback deadline is due. This work is different
from ours as we focus on caching contents to reduce backhaul traffic.

\section{Motivation\label{sec:motivation}}

\begin{figure}[t]
  \centering
  \subfloat[Time $\tau$]
  {
    \label{subfig:t1}
    \includegraphics[width=.2\textwidth]{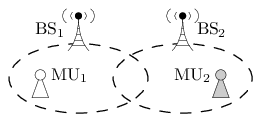}
  }
  \subfloat[Time $\tau+1$]
  {
    \label{subfig:t2}
    \includegraphics[width=.2\textwidth]{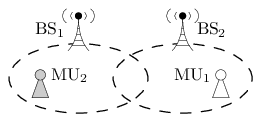}
  }
  \caption{A simple network with two small-BSs and two MUs\label{fig:topo}}
\end{figure}

To motivate the mobility-aware content caching problem, consider the
network shown in Fig.~\ref{fig:topo} with two small-BSs (BS$_1$ and
BS$_2$) and two MUs (MU$_1$ and MU$_2$). In time slot $\tau$, {MU}$_1$
is located within the cell of {BS}$_1$ and {MU}$_2$ within {BS}$_2$
(Fig.~\ref{subfig:t1}). In the next time slot $\tau+1$, both MUs move
to the other cell (Fig.~\ref{subfig:t2}). Both MUs make requests to a
library of three objects (O$_1$, O$_2$, and O$_3$). We assume that the
MUs' requests are recurrent, {\em i.e.}, an MU may request the same
object at different time instances according to the MU's preference
for the object. For instance, Spotify users who do not turn on the
``available offline'' option need to download a music each time the
music is played. If an MU requests an object that is not cached
locally, the corresponding small-BS needs to retrieve the object over
its backhaul network, which incurs a {\it cost} (such as the amount of
energy consumed at the small-BS or the amount of traffic generated in
the backhaul network). We term the {\it product} of the cost of
retrieving an object over a backhaul and the rate at which an MU
requests the object the {\it normalized cost} of the object with
respect to the MU. The normalized cost of an object thus represents
the expected cost incurred by an MU requesting the object in unit
time, if the object is not cached at an associated small-BS. 

\begin{table}[b]
  \centering 
  \begin{tabular}{c | c c c } 
    & {O}$_1$ & {O}$_2$ & {O}$_3$ \\\hline 
    {MU}$_1$ & 8 & 1 & 7\\ 
    {MU}$_2$ & 1 & 9 & 7\\ 
  \end{tabular} 
  \caption{Normalized costs of objects with respect to MU$_1$ and
    MU$_2$ \label{table:cost}} 
\end{table}

For instance, Table~\ref{table:cost} lists the normalized costs of the
three objects with respect to MU$_1$ and MU$_2$. In this example, we
assume that each small-BS has the storage capacity to cache only one
object and an MU's request to a cached object incurs zero cost. In
each time slot, the {\em network cost} is the sum of the normalized
costs of the requested but uncached objects. The {\em total} network
cost is the sum of the network costs incurred over the time slots of a
given time period.

Without considering the mobility of MUs, a small-BS would decide to
cache the object with the highest normalized cost of all the MUs
within its cell. In this example, BS$_1$ only considers the first row
in Table~\ref{table:cost} and BS$_2$ the second row.  Consequently,
{BS}$_1$ would choose to cache {O}$_1$ and {BS}$_2$ caches {O}$_2$.
Therefore, the network cost incurred in time slot $\tau$ is $(1+7) +
(1+7) = 16$, while the network cost in time slot $\tau+1$ is
$(8+7)+(9+7)=31$. The total network cost is $47$. 

In contrast, if a small-BS considers the mobility of MUs ({\it i.e.},
the requests from its future associated MUs) and decides to cache the
objects whose normalized costs are high with respect to both current
MUs and future MUs, the total network cost can be reduced. For
instance, by considering both rows in Table~\ref{table:cost}, both
{BS}$_1$ and BS$_2$ will unanimously cache O$_3$ and the total network
cost is reduced to $19+19=38$. The example clearly shows that taking
the mobility of MUs into consideration can greatly improve the caching
efficiency and reduce the backhaul cost incurred by those uncached
objects.

\section{Mobility-Aware Content Caching Problem\label{sec:framework}}

Given a library of contents, a group of MUs and a set of small-BSs
with storage of finite capacity, the mobility-aware content caching
problem concerns how the contents should be cached at the small-BSs to
minimize the total network cost incurred over the backhaul networks.
In this section, we first state the assumptions made in
Sec.~\ref{subsec:assumptions}, and then formulate the problem in
Sec.~\ref{subsec:form}.

\subsection{Assumptions\label{subsec:assumptions}}

We make the following assumptions.
\begin{itemize}[leftmargin=*]
  \item MUs' preferences for contents are known {\em a priori} (for
    instance, by analyzing MUs' request histories) and supplied as
    input to the problem. It is also assumed that MUs' preferences for
    contents remain constant ({\em i.e.}, do not change drastically
    over a period of time).
  \item Mobility patterns of MUs are also known {\em a priori}. For
    instance, MUs' mobility patterns may be specified in their
    profiles. Also, existing work
    \cite{nicholson2008breadcrumbs}\cite{pang2010wifi} may analyze
    MUs' historical information and predict their future connectivity
    with small-BSs. Our work focuses on utilizing either the specified
    or the predicted mobility information to facilitate effective
    content caching at small-BSs.
  \item Time is discretized into a sequence of time slots. An MU is
    assumed to be static within a single time slot, and may move into
    a different small cell at the beginning of the next time slot.
  \item The contents are of equal sizes and the constraint on cache
    capacity is specified in terms of the maximum number of contents
    that can be cached by a small-BS. The assumption of equal-sized
    contents leads to a cleaner formulation\footnote{This assumption
    can be lifted by introducing auxiliary binary variables that
  indicate whether the heterogeneously sized contents are cached at a
small-BS.}.
  \item When an uncached content is requested, the content has to be
    retrieved from a remote server over the backhaul. When the content
    arrives at a small-BS, one issue is whether the new content should
    replace an existing one in the cache. Such cache replacement is
    beyond the scope of this work, and we assume that, when the
    storage is full, a newly arrived content will be used once and
    then discarded, {\em i.e.}, the content will not replace any
    existing cached content. This is due to the fact that more
    preferred contents will cached earlier.
\end{itemize}

\subsection{Problem Formulation\label{subsec:form}}

Consider a small-cell network with a set $\mathcal{F}$ of small-BSs
and a set $\mathcal{I}$ of MUs that are interested in a set
$\mathcal{L}$ of contents. A small-BS $f$ is equipped with
a storage device that has the capacity to cache up to $C_f$ contents,
where $C_f$ is the cache capacity of $f$.  The set of contents cached
at small-BS $f$ is denoted as $\mathcal{L}_f$. Obviously
$\mathcal{L}_f \subseteq \mathcal{L}$ and $|\mathcal{L}_f|\le C_f$. 

Time is partitioned into a set $\mathcal{T} = \{1, 2,\cdots,
  |\mathcal{T}|\}$ of consecutive time slots. Given time slot $\tau$,
  $\mathcal{F}_i^\tau$ describes the set of small-BSs
  reachable\footnote{Small-BS $f$ is reachable from MU $i$ if MU $i$
  can sense the signal from small-BS $f$. Given the dense deployment,
multiple small-BSs may be reachable by an MU at a given time.} from MU
$i$. Due to mobility, each MU may see different sets of small-BSs at
different time slots, {\it i.e.} $\mathcal{F}_i^{\tau_1}$ may differ
from $\mathcal{F}_i^{\tau_2}$ if $\tau_1 \neq \tau_2$. 

At any time, MU $i$ incurs a normalized cost $c_{i,l}$ if content $l$
is not cached by any small-BS $f\in\mathcal{F}_i^\tau$. As mentioned
in Sec.~\ref{sec:motivation}, the normalized cost of a content with
respect to an MU is the product of the cost of retrieving the content
over the backhaul and the rate at which the content is requested by
the MU, which thus represents the expected cost of the MU requesting
the content in unit time. On the other hand, if at time slot $\tau$,
MU $i$ requests a content that is cached by some small-BS in
$\mathcal{F}_i^\tau$, zero cost is incurred since the request can be
satisfied by the small-BS without incurring any content retrieval over
the backhaul. Denoting $\mathcal{O}^\tau_i$ the set of cached contents
that are available to MU $i$ at time slot $\tau$, $\mathcal{O}^\tau_i$
is thus the union of the cached contents at the small-BSs reachable
from MU $i$, {\it i.e.}, $\mathcal{O}^\tau_i =
\cup_{f\in\mathcal{F}_i^\tau} \mathcal{L}_f$.  The total network cost
is defined as the sum of the normalized costs of the uncached contents
incurred by all the MUs over the set of time slots, which is expressed as
$\sum_{\tau\in\mathcal{T}} \sum_{i\in\mathcal{I}}
\sum_{l\in\mathcal{O}} c_{i, l} \cdot \ind{l \notin
  \mathcal{O}_\tau^i}$, where $\ind{x}$ is an indicator function that
  is 1 when $x$ is true, and 0 otherwise. We can then write the total
  network cost as 
\begin{align}
  & \sum_{\tau\in\mathcal{T}} \sum_{i\in\mathcal{I}}
  \sum_{l\in\mathcal{O}} c_{i, l} \cdot \ind{l \notin
    \mathcal{O}^\tau_i}\nonumber \\
  = & \sum_{\tau\in\mathcal{T}} \sum_{i\in\mathcal{I}}
  \sum_{l\in\mathcal{O}} c_{i, l} \cdot \left(1 - \ind{l \in
    \mathcal{O}^\tau_i}\right)\nonumber \\ 
  = & \sum_{\tau\in\mathcal{T}} \sum_{i\in\mathcal{I}}
  \sum_{l\in\mathcal{O}} c_{i,l} - \sum_{\tau\in\mathcal{T}}
  \sum_{i\in\mathcal{I}} \sum_{l\in\mathcal{O}^\tau_i} c_{i, l}.
  \label{eqn:rewritten_cost}
\end{align} 
Notice that the first item, {\em i.e.}, $\sum_{\tau\in\mathcal{T}}
\sum_{i\in\mathcal{I}} \sum_{l\in\mathcal{O}} c_{i,l}$ in
(\ref{eqn:rewritten_cost}), does not change over different caching
decisions, and hence can be regarded as a constant. Consequently,
minimizing the total network cost is equivalent to maximizing the
second item in (\ref{eqn:rewritten_cost}). We term the second item
{\em caching utility} $\mathcal{U}$ of a caching scheme. The caching
utility represents the cost {\em saved} by caching contents to avoid
retrieving the contents from remote servers over the backhual.

The mobility-aware content caching problem is formulated as
\begin{align}
  \label{eqn:obj} \max_{\{\mathcal{L}_f\}} & \quad \mathcal{U} =
  \sum_{\tau\in\mathcal{T}} \sum_{i\in\mathcal{I}}
  \sum_{l\in\mathcal{O}^\tau_i} c_{i, l}, & \\
  \mbox{s.t.} &\quad \mathcal{L}_f\subseteq \mathcal{L}, \quad
  |\mathcal{L}_f| \le C_f, & \quad \forall f \in\mathcal{F},
  \label{eqn:main:capa} \\ 
  &\quad \mathcal{O}^\tau_i = \cup_{f\in\mathcal{F}_i^\tau}
  \mathcal{L}_f, & \forall i\in\mathcal{I}, \quad \forall
  \tau\in\mathcal{T}. \label{eqn:set}
\end{align}
The objective in (\ref{eqn:obj}) maximizes the caching utility of the
caching scheme. The constraint in (\ref{eqn:main:capa}) ensures that
the number of cached contents at each small-BS does not exceed the
cache capacity $C_f$, and (\ref{eqn:set}) specifies the set of cached
contents that are reachable from MU $i$ at time $\tau$.


We can prove the problem defined in (\ref{eqn:obj})-(\ref{eqn:set}) to
be NP-complete by reducing the problem defined in
\cite{femtocaching2012golrezaei} to our mobility-aware content
caching problem. In fact, the problem defined in
\cite{femtocaching2012golrezaei} is a special instance of our problem
with one single time slot. We omit the proof due to space limitation.

\section{Heuristic Solution\label{sec:heuristic}}

As the problem formulated in (\ref{eqn:obj})-(\ref{eqn:set}) is
NP-complete, in this section, we describe a heuristic solution termed
{\em MobiCacher}. We will prove that the worst-case performance of
{\em MobiCacher} is bounded by an approximation ratio $F$, where $F$
is the maximum number of small-BSs that an MU can sense at any time.
We also evaluate the performance of {\em MobiCacher} via trace-based
simulation, and present the results in Sec.~\ref{sec:results}.

The complexity of the problem defined in
(\ref{eqn:obj})-(\ref{eqn:set}) is largely due to the fact that an MU
may access the caches of multiple small-BSs at a time, and thus,
deciding whether a content should be cached at one small-BS is
affected by what contents are cached by nearby small-BSs.  To simplify
the decision making, {\em MobiCacher} decomposes the original problem
into $|\mathcal{F}|$ sub-problems, with one sub-problem for each
small-BS. Each sub-problem is solved independently, which enables {\em
MobiCacher} to run in polynomial time.

Let $\mathcal{I}^\tau_f$ denote the set of MUs that are associated
with small-BS $f$ at time slot $\tau$. The sub-problem for small-BS
$f$ tries to select a set of contents $\mathcal{L}_f$ that maximizes
the caching utility generated by MUs in $\cup_{\tau\in\mathcal{T}}
\mathcal{I}^\tau_f$. This sub-problem for small-BS $f$ is formulated
as
\begin{align}
  \max_{\mathcal{L}_f} & \quad \mathcal{U}_f =
  \sum_{l\in\mathcal{L}_f} \sum_{\tau\in\mathcal{T}}
  \sum_{i\in\mathcal{I}^\tau_f} c_{i,l} \label{eqn:new:obj}\\
  \mbox{s.t.} &\quad \mathcal{L}_f \subseteq \mathcal{L}, \quad
  |\mathcal{L}_f| \le C_f \label{eqn:new:const}
\end{align}

Notice that $\sum_{\tau\in\mathcal{T}} \sum_{i\in\mathcal{I}_f^\tau}
c_{i, l}$ can be rewritten as $\sum_{i\in\mathcal{I}} T_i^f c_{i,l}$,
where $T_i^f$ is the total number of time slots  that MU $i$ is
associated with small-BS $f$ (termed MU $i$'s {\it sojourn time} in
small-BS $f$). Thus, to solve the problem defined in
(\ref{eqn:new:obj})-(\ref{eqn:new:const}), small-BS $f$ can
iteratively select content $l$ with maximal $\sum_{i\in\mathcal{I}}
T_i^f c_{i,l}$, until the cache at small-BS $f$ fills up. This process
is summarized in Algorithm~\ref{algo}. The complexity of this
algorithm is of polynomial time $O(C_f \cdot |\mathcal{I}|)$. We now
analyze the approximation ratio of {\em MobiCacher} by proving the
following theorem.

\begin{theorem}
  {\em MobiCacher} is a polynomial-time $F$-approximation
  algorithm, where $F$ is the maximum number of small-BSs that an MU
  can sense at any time.\label{theorem:approx}
\end{theorem}

\begin{proof} 
  Let $\mathcal{U}'$ denote the caching utility achieved by {\em
  MobiCacher} and $\{\mathcal{L}'_f\}$ the corresponding sets of
  contents selected to be cached at small-BSs. By the definition of
  caching utility in (\ref{eqn:obj}), $\mathcal{U}'$ can be expressed
  in terms of $\{\mathcal{L}'_f\}$ as
  \begin{align*}
    \mathcal{U}' = \sum_{\tau\in\mathcal{T}} \sum_{i\in\mathcal{I}}
    \sum_{l\in\mathcal{O'}_i^\tau} c_{i,l} = \sum_{\tau\in\mathcal{T}}
    \sum_{i\in\mathcal{I}}
    \left\{\sum_{l\in\cup_{f\in\mathcal{F}_i^\tau} \mathcal{L}'_f}
    c_{i,l}\right\}.
  \end{align*} 
  We introduce binary variable $B'_{l,f}$ to indicate whether
  content $l$ is selected by {\em MobiCacher} to be cached at small-BS
  $f$. The membership of each set $\mathcal{L}'_f$ can be fully
  expressed by the set $\{B'_{l,f}\}$ of binary variables, and we
  rewrite caching utility $\mathcal{U}'$ as
  \begin{align} 
    \mathcal{U}' = \sum_{\tau\in\mathcal{T}} \sum_{i\in\mathcal{I}}
    \sum_{l\in\mathcal{O}} \left\{ c_{i,l} \cdot \min\left(1,
    \sum_{f\in\mathcal{F}_i^\tau} B'_{l,f} \right)\right\}.
    \label{eqn:mobicache_utility} 
  \end{align}

  Similarly, we denote the optimal caching utility by $\mathcal{U}^*$
  and use $\{B^*_{l,f}\}$ to represent the cached contents attained
  by the optimal solution. 
  To show that {\em MobiCacher} is an $F$-approximation algorithm, we
  need to show that $\frac{\mathcal{U}^*}{\mathcal{U}'} \le F$.
  The remainder of the proof continues from
  (\ref{eqn:mobicache_utility}).
  \begin{align}
    \mathcal{U}' \ge \sum_{\tau\in\mathcal{T}} \sum_{i\in\mathcal{I}}
    \sum_{l\in\mathcal{O}} \left\{ c_{i,l} \cdot
    \frac{1}{|\mathcal{F}^\tau_i|} \sum_{f\in\mathcal{F}_i^\tau}
    \min\left(1, B'_{l,f} \right)\right\} \label{eqn:nine} 
  \end{align}
  Eq.~(\ref{eqn:nine}) holds as $\min\left(1, \sum_{i=1}^N a_i\right) \ge
  \frac{1}{N} \sum_{i=1}^N \min\left(1, a_i\right)$, where each $a_i$
  only takes a binary value. We have defined $F$ as the maximum number
  of small-BSs that an MU can sense at any time, {\em i.e.}, $F \ge
  |\mathcal{F}_i^\tau|, \forall i\in\mathcal{I}, \forall
  \tau\in\mathcal{T}$. Thus we have
  \begin{align}
    \mathcal{U}' & \ge \sum_{\tau\in\mathcal{T}}
    \sum_{i\in\mathcal{I}} \sum_{l\in\mathcal{O}} \left\{ c_{i,l}
    \cdot \frac{1}{F} \sum_{f\in\mathcal{F}_i^\tau} \min\left(1,
    B'_{l,f} \right)\right\} \nonumber \\ & = \frac{1}{F}
    \sum_{f\in\mathcal{F}} \sum_{\tau\in\mathcal{T}}
    \sum_{i\in\mathcal{I}} \sum_{l\in\mathcal{O}} \left\{ c_{i,l}
    \cdot \min\left(1, B'_{l,f} \right)\right\} \label{eqn:ten}
  \end{align}
  Notice that $\sum_{\tau\in\mathcal{T}} \sum_{i\in\mathcal{I}}
  \sum_{l\in\mathcal{O}} \left\{ c_{i,l} \cdot \min\left(1, B'_{l,f}
\right)\right\}$ in (\ref{eqn:ten}) is the objective function
in (\ref{eqn:new:obj}) rewritten with the binary variables
$\{B'_{l,f}\}$. Since the value of each $B'_{l,f}$ is chosen to
maximize (\ref{eqn:new:obj}) as opposed to (\ref{eqn:obj}), the
following inequality holds.
  \begin{eqnarray}
    & \sum_{\tau\in\mathcal{T}} \sum_{i\in\mathcal{I}}
    \sum_{l\in\mathcal{O}} \left\{ c_{i,l} \cdot \min\left(1, B'_{l,f}
  \right)\right\} \ge\nonumber \\ 
    &\sum_{\tau\in\mathcal{T}} \sum_{i\in\mathcal{I}}
    \sum_{l\in\mathcal{O}} \left\{ c_{i,l} \cdot \min\left(1,
    B^*_{l,f} \right)\right\}, \quad \forall f\in\mathcal{F} 
    \nonumber \\ 
    \label{eqn:eleven}
  \end{eqnarray} 
  Combining (\ref{eqn:ten}) and (\ref{eqn:eleven}) we have
  \begin{align}
    \mathcal{U}' & \ge \frac{1}{F} \sum_{f\in\mathcal{F}}
    \sum_{\tau\in\mathcal{T}} \sum_{i\in\mathcal{I}}
    \sum_{l\in\mathcal{O}} \left\{ c_{i,l} \cdot \min\left(1,
    B^*_{l,f} \right)\right\} \nonumber \\
    & = \frac{1}{F} \sum_{\tau\in\mathcal{T}} \sum_{i\in\mathcal{I}}
    \sum_{l\in\mathcal{O}} \left\{ c_{i,l} \cdot
    \sum_{f\in\mathcal{F}_i^\tau} \min\left(1, B^*_{l,f}
  \right)\right\} \nonumber \\
    & \ge \frac{1}{F} \sum_{\tau\in\mathcal{T}} \sum_{i\in\mathcal{I}}
    \sum_{l\in\mathcal{O}} \left\{ c_{i,l} \cdot \min\left(1,
    \sum_{f\in\mathcal{F}_i^\tau} B^*_{l,f} \right)\right\}
    \label{eqn:twelve} 
  \end{align} 
  The last inequality in (\ref{eqn:twelve}) holds due to the fact that
  $\sum_{i=1}^N \min\left(1, a_i\right) \ge \min\left(1, \sum_{i=1}^N
  a_i\right)$, where each $a_i$ is a binary variable.  Combining
  (\ref{eqn:twelve}) with the definition of caching utility, it is
  readily seen that $\mathcal{U}' \ge \frac{1}{F} \mathcal{U}^*$, {\em
  i.e.}, $\frac{\mathcal{U}^*}{\mathcal{U}'} \le F$. 
\end{proof}

\begin{algorithm}[t]
  \caption{{\em MobiCacher} for small-BS $f$\label{algo}}
  \begin{algorithmic}[1]
    \REQUIRE A set of candidate contents $\mathcal{L}$; 
    a set of MUs $\mathcal{I}$; 
    MU preferences $\{c_{i, l}\}$; 
    MU sojourn times $\{T_i^f\}$.
  \STATE $\mathcal{L}_f \leftarrow \emptyset$
  \WHILE{$|\mathcal{L}_f| < C_f$}
    \STATE $l\leftarrow\argmax_l{\sum_{i\in\mathcal{I}} T_i^f c_{i,l}}$
    \STATE $\mathcal{L}_f \leftarrow \mathcal{L}_f \cup \{l\}$
    \STATE $\mathcal{L} \leftarrow \mathcal{L} \setminus \{l\}$
  \ENDWHILE
  \RETURN $\mathcal{L}_f$
\end{algorithmic}
\end{algorithm}

\section{Performance Evaluation\label{sec:results}}

\subsection{User Mobility Model\label{subsec:mobility}}

We evaluate the performance of {\em MobiCacher} by adapting the real
trace of mobile users collected by the Wireless Topology Discovery
(WTD) project at UCSD \cite{mcnett2005access}. In the WTD trace, an
active device records all the access points (APs) that the device
could sense across all frequencies at each sampling point. The time
interval between two adjacent sampling points is 20 seconds. We thus
set the length of a time slot to be 20 seconds in our simulations. We
choose to use the WTD trace for our simulation because (1) the dataset
is publicly available, and (2) the WTD trace records multiple APs
visible to an MU at each sampling point, while most of the other traces
only record the AP that an MU is associated with at any time. Our use
of the WTD trace in the simulation substitutes small-BSs for WiFi APs.

The entire WTD trace includes data from 275 PDA devices accessing more
than 400 unique APs during the time interval from 2002/9/22 to
2002/12/08. We choose the busiest day of the trace, namely the day of
2002/10/16, during which the highest number of active MUs was
recorded.  In order to evaluate {\em MobiCacher} under different MU
mobility patterns, we run the simulation over four different one-hour
time intervals starting at 00:00:00, 06:00:00, 12:00:00 and 18:00:00.
Fig.~\ref{fig:cumulative} plots the cumulative 
distribution of MUs' sojourn times in small-BSs in these four time
intervals. We can see that MUs are more mobile (with shorter sojourn
time) during the daytime than that during the night time.

\subsection{User Preference Model}

To model MUs' preferences for contents, we use the song listening logs
collected by Last.fm \cite{thierry2011the}. The Last.fm dataset
contains the song listening history of about 1000 unique MUs. We
select the 200 top-played songs and put the songs into the library
$\mathcal{L}$ of contents. We rank a MU's preference for a
song in $\mathcal{L}$ as how many times the MU listened to the song
divided by the total number of songs the MU listened to. For example,
if an MU listened to a particular song 5 times and the total number of
songs the MU listened to is 100, then the MU's preference for this
song is 0.05.  Moreover, if the MU can retrieve the song from one of
the caches of reachable small-BSs in a time slot, the MU gains a
utility of 0.05. To evaluate the performance of {\em MobiCacher}, we
randomly associate an MU in the real traces generated in
Sec.~\ref{subsec:mobility} with a listening profile and calculate the
MU's preference for songs using the rule mentioned above.

\begin{figure}
  \centering
  \includegraphics[width=\simfigwidth]{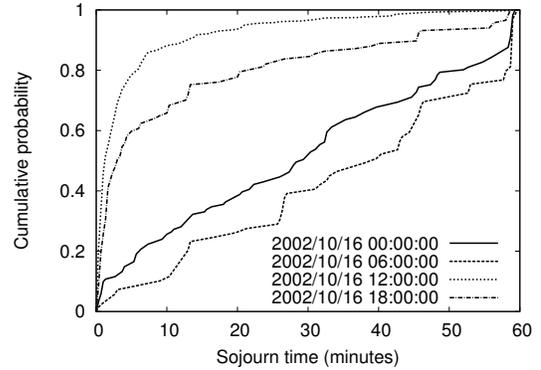}
  \caption{Cumulative distribution of the MUs' sojourn
    times\label{fig:cumulative}}
\end{figure}

\begin{figure*}
  \centering
  \subfloat[2002/10/16 00:00:00\label{subfig:00}] {
    \includegraphics[width=\smallsimfigwidth]{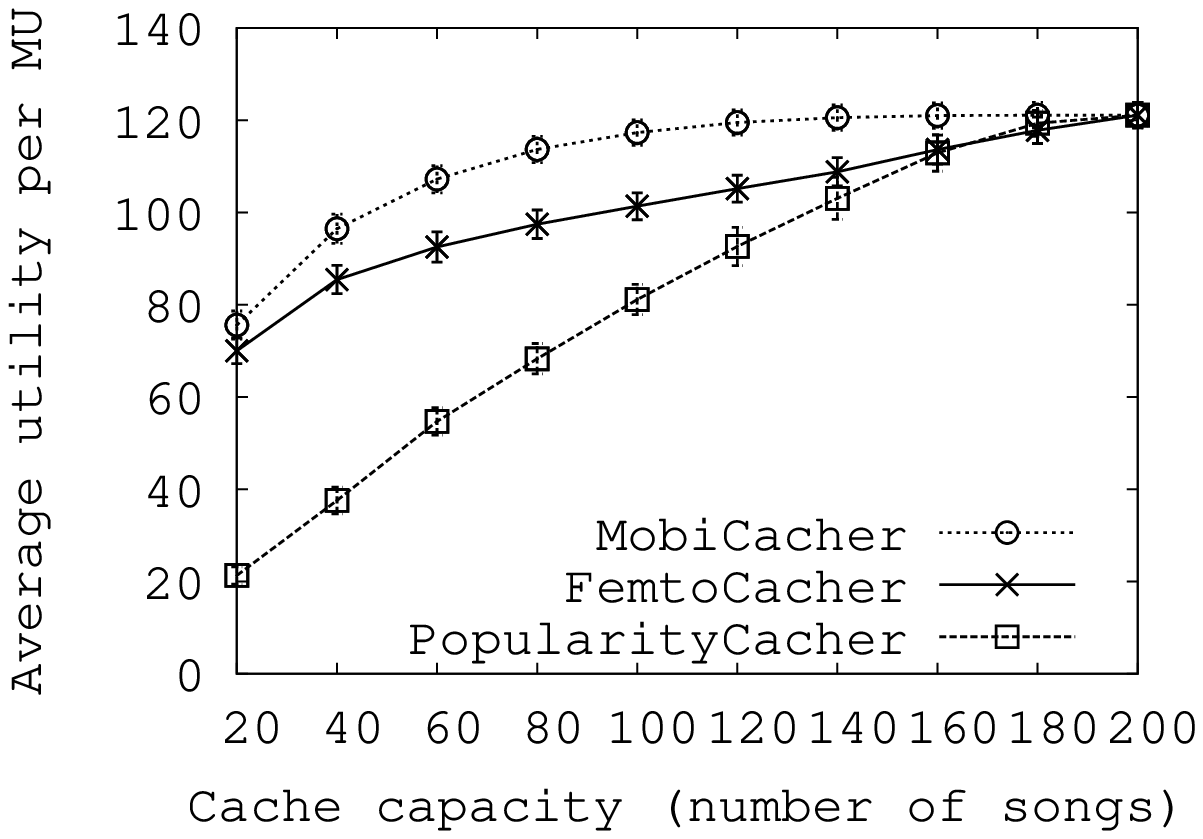}
  }
  \subfloat[2002/10/16 06:00:00\label{subfig:06}] {
    \includegraphics[width=\smallsimfigwidth]{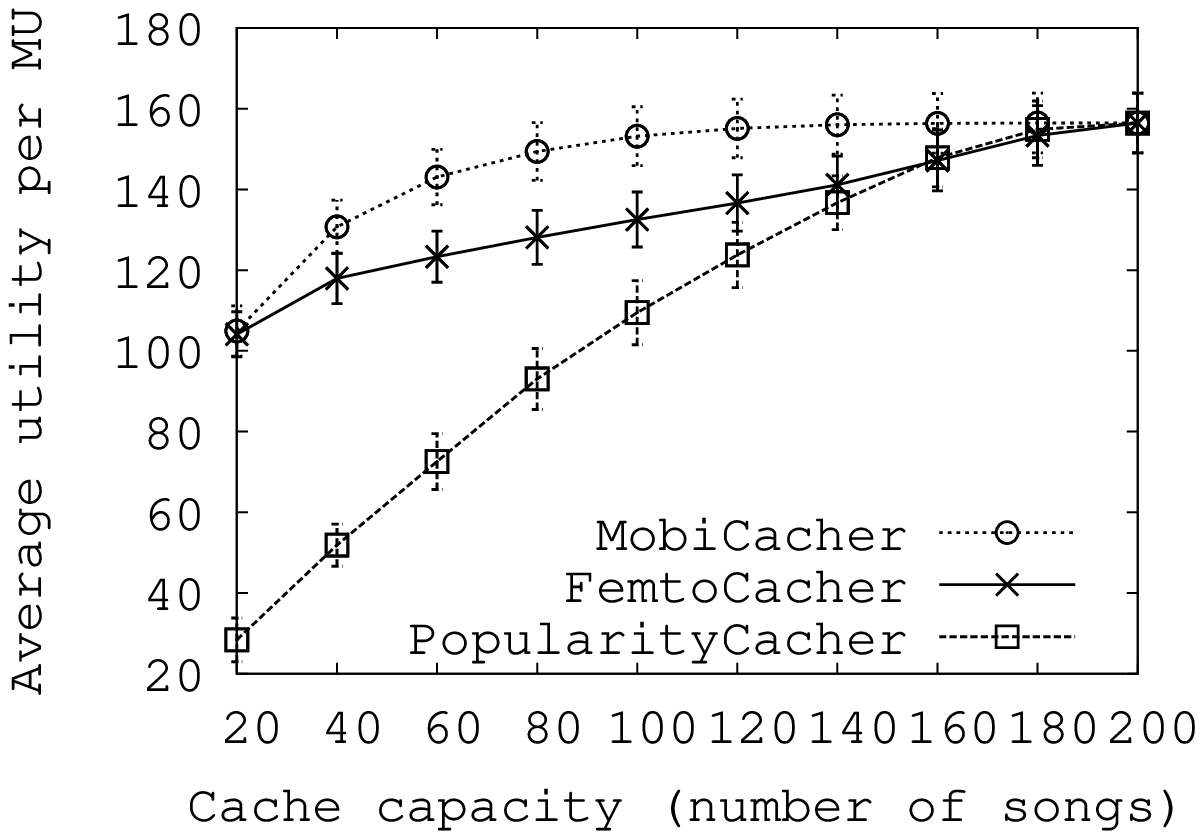}
  }
  \subfloat[2002/10/16 12:00:00\label{subfig:12}] {
    \includegraphics[width=\smallsimfigwidth]{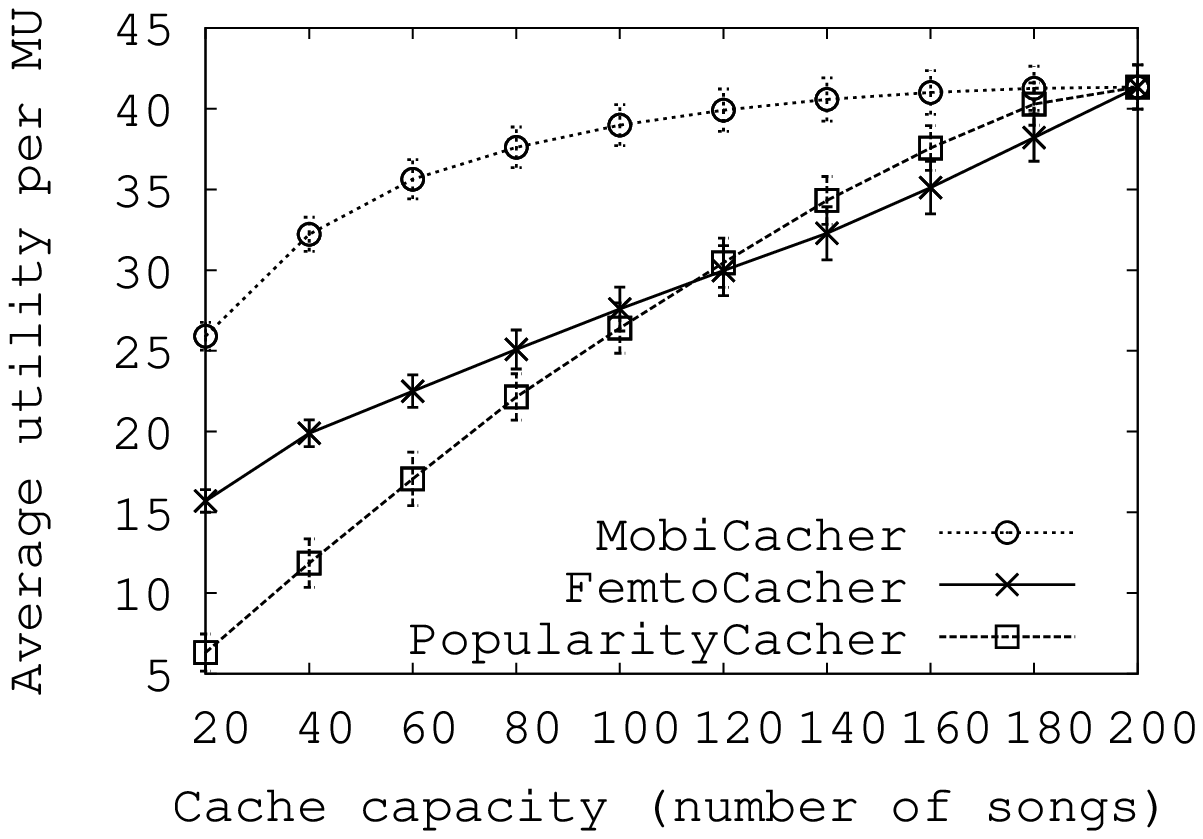}
  }
  \subfloat[2002/10/16 18:00:00\label{subfig:18}] {
    \includegraphics[width=\smallsimfigwidth]{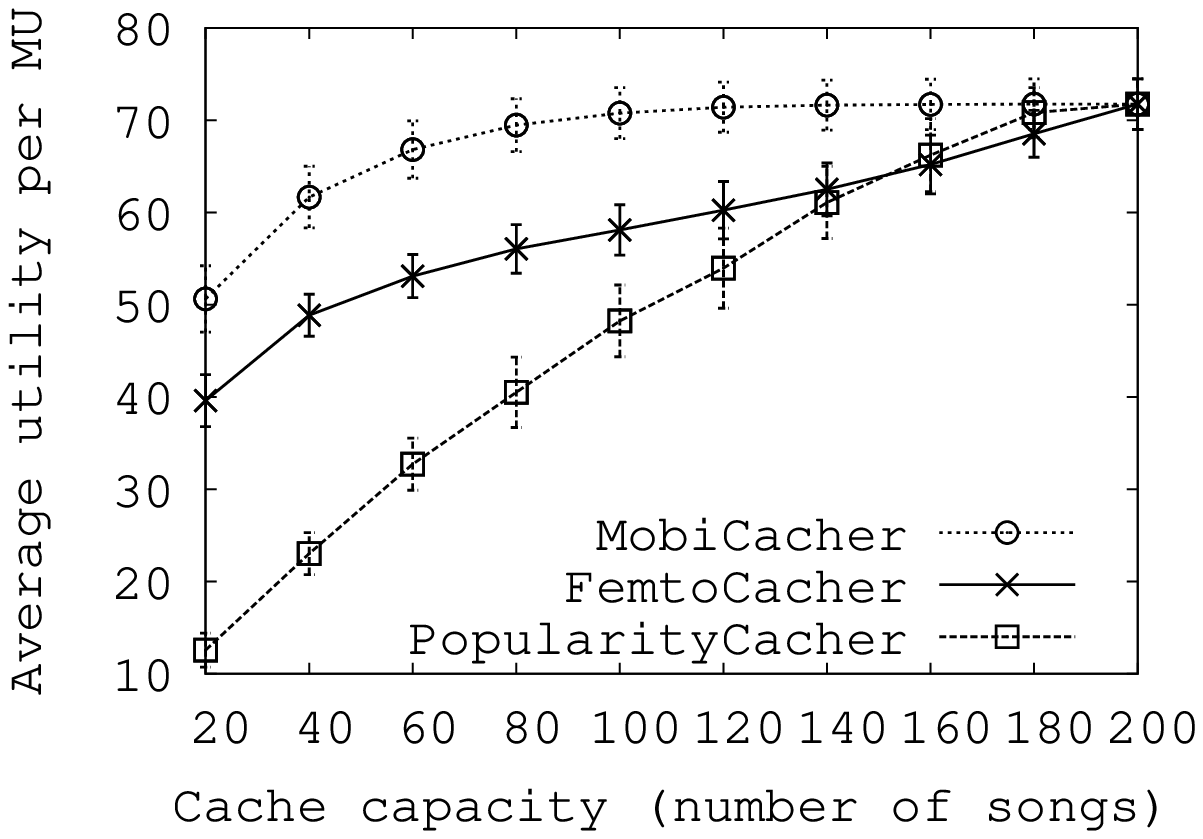}
  }
  \caption{Impact of cache capacity at each small-BS\label{fig:total_utility}}
\end{figure*}

\subsection{Performance of {\em MobiCacher}}

We compare the performance of three schemes:
(1) {\em MobiCacher} -- our mobility-aware caching solution,
(2) {\em FemtoCacher} -- the caching solution proposed in
\cite{femtocaching2012golrezaei} where MUs' mobilities are not
considered and the caching decision is made based on the geographical
distribution of MUs at the beginning of the simulations,
and (3) {\em PopularityCacher} -- each small-BS caches the most
popular songs.


Fig.~\ref{fig:total_utility} plots the average caching utility gained
by each MU against the cache capacity of each small-BS. The results
clearly show that the caching utility produced by {\em MobiCacher}
over the four one-hour simulation periods is higher than the 
utility generated by {\em FemtoCacher} and {\em PopularityCacher}. In
particular, the performance improvement of {\em MobiCacher} is more
significant  when the cache capacity is small, {\em i.e.}, when each
small-BS can cache less than 60 songs. This is due to the fact that,
when the storage capacity of each cache is small, {\em MobiCacher} can
make more informed caching decisions by considering the mobility of
MUs. On the contrary, when the cache capacity is large ({\em i.e.},
each small-BS can cache up to 160 songs), the three caching schemes
perform similarly as almost the entire song library can be cached at
each small-BS.

{\em PopularityCacher} produces the least utility compared with that
of {\em MobiCacher} and {\em FemtoCacher}, as {\em PopularityCacher}
only considers the popularity of the songs but not the geographical
distribution of MUs nor the mobilities of MUs. Also, a song that is
the most popular globally may not be the most popular among the MUs
associated with a particular small-BS. As all of the small-BSs store
the most popular songs, the contents cached in the small-BSs at
different locations are not diversified. In other words, the fact that
an MU may access multiple small-BSs at a time is not taken advantage
of as the cached songs at each small-BS are the same.

Both {\em MobiCacher} and {\em FemtoCacher} consider the geographical
distribution of MUs. {\em FemtoCacher} assumes that MUs are static and
only looks at the initial MU distribution. In practice, an MU
may move among the cells of small-BSs. As time progresses, the current
geographical distribution of the MUs can be quite different from their
initial distribution. {\em FemtoCacher} fails to exploit this fact.
The influence of MU mobility is shown in Fig.~\ref{fig:slot} where we
plot the cumulative utility against the simulation time. We can see
that {\em FemtoCacher} performs similarly or even slightly better than
{\em MobiCacher} at the beginning of the simulation. But as time
progresses and MUs start moving to different locations, the cumulative
utility of {\em FemtoCacher} lags behind that of {\em MobiCacher}. The
impact of MUs' mobility can also be observed from
Figs.~\ref{fig:cumulative}-\ref{fig:total_utility}. When MUs
are less mobile ({\em i.e.}, at night at 00:00:00 or in the early
morning at 06:00:00), the gap between {\em MobiCacher} and {\em
FemtoCacher} is relatively small. However, during the daytime when MUs
are highly mobile, {\em MobiCacher} outperforms {\em FemtoCacher} by
up to 27\% as shown in
Figs.~\ref{fig:total_utility}(c)-\ref{fig:total_utility}(d).

\begin{figure}[t]
  \centering
  \includegraphics[width=\simfigwidth]{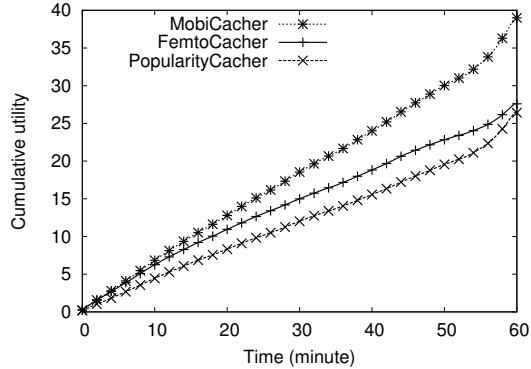}
  \caption{Cumulative utility against time (cache capacity: 100 songs,
    simulation start time: 2002/10/16 12:00:00) \label{fig:slot}}
\end{figure}

\section{Summary and Future Work\label{sec:conclusion}}

The paper defines a mobility-aware content caching problem, which
caches popular contents at the base stations of small-cell network to
ease the stress on the backhaul networks. We formulate the problem
into an optimization framework with the objective of maximizing the
utility of caching.  Given the NP-completeness of the problem, we
developed the {\em MobiCacher} heuristic solution. We prove that the
approximation ratio of {\em MobiCacher} is bounded and evaluate the
performance via trace-based simulation.  Results show that {\em
MobiCacher} yields higher caching utility than other caching
solutions.

Throughout the discussion of this paper, we assume that the
trajectories of MUs are either known a priori or can be predicted.  In
practice, trajectory prediction might not be perfect and the mobility
patterns of MUs can be random. We plan to investigate a {\em robust}
version of the mobility-aware content caching problem by explicitly
taking these uncertainties into the problem formulation.




{The views and conclusions contained herein are those of the authors and
should not be interpreted as necessarily representing the official policies or
endorsements, either expressed or implied, of the Air Force Research
Laboratory or the U.S. Government.}

\end{document}